\documentclass[a4paper,UKenglish]{lipics-v2018-nopage}
\usepackage{microtype}
\usepackage{tikz}
\usepackage{verbatim}
\usepackage{tkz-graph}
\usepackage{mathabx} 
\usepackage{enumitem} 
\usetikzlibrary{positioning,decorations.pathreplacing}

\newcommand{\ket}[1]{\vert #1 \rangle}

\newcommand{\mybar}[1]{\lambda}

\newcommand{\Gg}{\mathscr{G}}

\newcommand{\Pp}{\mathscr{P}}
\newcommand{\suppP}{\Lambda}
\newcommand{\distrP}{\Gamma}

\newenvironment{proof-sketch}{\trivlist\item[]\emph{Brief proof sketch}:}%
{\unskip\nobreak\hskip 1em plus 1fil\nobreak$\Box$
\parfillskip=0pt%
\endtrivlist}
\theoremstyle{plain}
\newtheorem{proposition}[theorem]{Proposition}
\newtheorem{claim}[theorem]{Claim}
\newtheorem{process}{Process}

\title{Quantum Advantage for the LOCAL Model in Distributed Computing}
\titlerunning{Quantum Advantage for the LOCAL Model in Distributed Computing}

\author{Fran{\c c}ois Le Gall}{Graduate School of Informatics, Kyoto University\\
Yoshida-Honmachi, Sakyo-ku, Kyoto 606-8501, Japan
}{}{}{}

\author{Harumichi Nishimura}{Graduate School of Informatics, Nagoya University\\
Chikusa-ku, Nagoya, Aichi 464-8601, Japan 
}{}{}{}

\author{Ansis Rosmanis}{Centre for Quantum Technologies, National University of Singapore\\
Block S15, 3 Science Drive 2, 117543, Singapore
}{}{}{}

\authorrunning{F. Le Gall, H. Nishimura and A. Rosmanis}

\Copyright{Fran{\c c}ois Le Gall, Harumichi Nishimura and Ansis Rosmanis
}

\subjclass{}


\category{}

\relatedversion{}

\supplement{}

\funding{}

\nolinenumbers 
\hideLIPIcs  

\begin{document}

\maketitle

\begin{abstract}
There are two central models considered in (fault-free synchronous) distributed computing: the CONGEST model, in which communication channels have limited bandwidth, and the LOCAL model, in which communication channels have unlimited bandwidth. Very recently, Le Gall and Magniez (PODC 2018) showed the superiority of quantum distributed computing over classical distributed computing in the CONGEST model. In this work we show the superiority of quantum distributed computing in the LOCAL model: we exhibit a computational task that can be solved in a constant number of rounds in the quantum setting but requires $\Omega(n)$ rounds in the classical setting, where $n$ denotes the size of the network.
 \end{abstract}

\section{Introduction}

\noindent
{\bf Classical distributed computing.} A central topic in distributed computing is the study of synchronous network algorithms. Here processors and communication channels are modeled using nodes and edges, respectively, and executions proceed with round-based synchrony, where each node can transfer one message to each adjacent node per round. The main quantity of interest is typically the number of rounds needed to solve a computational task. Two fundamental models considered in the literature are the LOCAL model, introduced by Linial \cite{LinialFOCS87,LinialSICOMP92}, and the CONGEST model, introduced by Peleg \cite{Peleg00}. 

The LOCAL model does not put any limitation on the size of the messages sent at each round, and thus mainly characterizes the locality of the problem considered and the hardness of breaking symmetry between nodes. Obviously all computational problems can be solved with $O(D)$ rounds in the LOCAL model, where $D$ is the diameter of the network, by first collecting all the information about the network (including the inputs of all nodes) at some node. Many important problems have significantly more efficient algorithms---we refer to~\cite{Peleg00} for examples and to \cite{Chang+FOCS17} for a recent classification. 

In the CONGEST model, on the other hand, each message has restricted length (the length is typically restricted to $O(\log n)$ bits, where $n$ is the number of nodes in the network). This corresponds to the situation of communication channels with limited bandwidth, in which case congestions can arise. A simple example showing the striking difference between these two models is deciding whether the diameter of the network is 2 or 3. This problem requires $\Theta(n)$ rounds in the CONGEST model \cite{Frischknecht+SODA12,Holzer+PODC12,Peleg+ICALP12}, while in the LOCAL model it can be trivially solved with a constant number of rounds.\vspace{2mm}

\noindent
{\bf Quantum distributed computing.} 
Quantum versions of both models can be naturally defined by replacing classical channels by quantum channels between the processors (which are now quantum processors, i.e., processors that can process quantum information). Gavoille, Kosowski and Markiewicz~\cite{Gavoille+DISC09} first considered quantum distributed computing in the LOCAL model, and showed that for several fundamental problems, such as Graph Coloring or Maximal Independent Set, allowing quantum communication cannot lead to any significant advantage. More recently, Arfaoui and Fraigniaud \cite{Arfaoui+14} observed that several lower bound techniques for the classical LOCAL model hold in the quantum model as well. 

The power of distributed network computation in the CONGEST model, where each node can send $O(\log n)$ qubits per round to each neighbor, has been first investigated by Elkin, Klauck, Nanongkai and Pandurangan~\cite{Elkin+PODC14}. The main conclusions reached in that paper were that for many fundamental problems in distributed computing, such as computing minimum spanning trees or minimum cuts, quantum communication does not, again, offer significant advantages over classical communication. Recently, Le Gall and Magniez nevertheless showed the superiority of quantum distributed computing in the CONGEST model for a concrete problem \cite{LeGall+PODC18}: they showed that the diameter of the network can be computed in $\tilde O(\sqrt{nD})$ rounds in the quantum setting, where $n$ is the number of nodes of the network and $D$ is the diameter of the network. In comparison, as mentioned above $\Omega(n)$ rounds are necessary in the classical setting even if $D$ is constant.  It should be mentioned that from a purely complexity-theoretic perspective most known separations between two-party classical and quantum communication complexities (e.g., separations in the bounded-error setting for the disjointness function \cite{Aaronson+05,Buhrman+STOC98, Hoyer+STACS02}) can be converted in a straightforward way into similar separations in the CONGEST model. The contribution of \cite{LeGall+PODC18} is actually to give a separation for an important problem in distributed computing.

A pressing open question is to understand whether a similar quantum speedup in distributed computing can be obtained in the LOCAL model. The only known gap is a factor~2: for each integer $t\ge 1$, Gavoille, Kosowski and Markiewicz~\cite{Gavoille+DISC09} constructed a computational problem (inspired by the work by Greenberger, Horne and Zeilinger \cite{Greenberger+89}) that can be solved in~$t$ rounds in the quantum setting but requires $2t$ rounds in the classical setting.\footnote{A much larger gap is shown in~\cite{Gavoille+DISC09} for the setting where the nodes of the network initially share a globally entangled state. In the present paper, however, we consider the arguably more natural setting where no prior entanglement is allowed.} The quantum upper bound comes from the observation that $t$ rounds are enough to create entanglement between two nodes at distance $2t$ from each other. In this perspective, as mentioned in \cite{Gavoille+DISC09}, the speed-up factor of 2 may ``look like a natural limit''. Note that, contrary to the CONGEST model, known separations between two-party (or multiparty) quantum and classical communication complexities seem meaningless to prove separations in the LOCAL model due to the unlimited bandwidth between nodes. \vspace{2mm}

\noindent
{\bf Our results.}
In this work we show the existence of a large gap between the round complexities of quantum and classical distributed computation in the LOCAL model.
\begin{theorem}\label{th1}
There exists a computational problem that can be solved with 2 rounds in the quantum LOCAL model, but requires $\Omega(n)$ rounds in the classical LOCAL model, where~$n$ denotes the number of nodes in the network. The classical lower bound holds even if arbitrary prior randomness is allowed.
\end{theorem}
The computational problem we construct to prove Theorem \ref{th1} is inspired by a construction from \cite{Barrett+PRA07}, which was initially used to show the non-locality of measurement outcomes of graph states. The same construction was recently also used by Bravyi, Gosset and K\"onig~\cite{Bravyi+17} to prove their separation between quantum and classical constant-depth circuit complexities. The problem, defined in Section \ref{sec:RelSep}, can be informally described as follows: on an $n$-node ring, the nodes should output one of the possible outcomes that arise when measuring the graph state corresponding to the ring in a basis depending on the input each node receives. We are currently not aware of any applications of Theorem \ref{th1} for constructing quantum algorithms for problems of interest to the distributed computing community, but nevertheless consider this result as a valuable proof of concept showing that the quantum LOCAL model can be arbitrarily more powerful than the classical LOCAL model. 

The computational problem considered in Theorem \ref{th1} is a relation (i.e., for each input there are multiple valid outputs). It is fairly easy to show that for any function (i.e., for each input there is only one valid output at each node) the quantum and classical round complexities are equal in the LOCAL model: we give a proof of this property in Appendix~\ref{sec:functions}. We then investigate whether a separation similar to the separation of Theorem~\ref{th1} can be obtained for a computational problem without input. Such kinds of computational problems (seen as sampling problems or computations of probability distributions) are the main targets of the field of quantum supremacy (see \cite{Arrow+Nature17} for a recent survey). Indeed, a major open problem left in the work by Bravyi, Gosset and K\"onig \cite{Bravyi+17} mentioned above is to prove the superiority of constant-depth quantum circuits for the computation of a probability distribution. We show that in the LOCAL model of distributed computing such a goal can be achieved.  

\begin{theorem}\label{th2}
There exists a sampling problem that can be solved with 2 rounds in the quantum LOCAL model, but requires $\Omega(n)$ rounds in the classical LOCAL model. The classical lower bound holds even for constant-error additive approximation.
\end{theorem}
Theorem \ref{th2} is proved by considering the same computational problem as used in Theorem~\ref{th1} but replacing the inputs by random bits. The proof nevertheless requires several adjustments, in particular a careful analysis of the classical randomness shared during the execution of the protocol. \vspace{2mm}

\noindent
{\bf Other relevant works.} It is well known that quantum communication can offer significant advantages over classical communication in several settings such as communication complexity or quantum games (see, e.g.,~\cite{Broadbent+08,Denchev+08,deWolf02}). Concerning problems of interest to the distributed computing community, the main works not already mentioned are quantum algorithms for byzantine agreements \cite{Ben-Or+STOC05} and for distributed computing over anonymous networks, and in particular the design of zero-error quantum algorithms for leader election \cite{Tani+12} (see also \cite{Denchev+08}).


\section{Preliminaries}\label{sec:prelim}
\subsection{Notations and definitions}\label{subsec:notation}
{\bf Quantum gates.}
We assume that the reader is familiar with the basis of quantum computation and refer to~\cite{Nielsen+00} for a standard reference.
We will use the Hadamard gate $\mathrm{H}$ and the phase gate $\mathrm{S}$ acting on one qubit: 
\[
\mathrm{H}=
\frac{1}{\sqrt{2}}\left(
\begin{array}{cc}
1&1\\
1&-1
\end{array}
\right),
\hspace{3mm}
\mathrm{S}=\left(
\begin{array}{cc}
1&0\\
0&i
\end{array}
\right),
\]
where $i$ denotes the imaginary unit of complex numbers. We will also use the $\mathrm{CNOT}$ gate acting on two qubits (called the control qubit and the target qubit) that maps the basis state $\ket{a}\ket{b}$, for any two bits $a,b\in\{0,1\}$, to the state $\ket{a}\ket{a\oplus b}$ where $\oplus$ denotes the exclusive OR. 
Finally, we will need the following two 2-qubit gates (Controlled-Z and Controlled-S gates):
\[
\mathrm{CZ}=\left(
\begin{array}{cccc}
1&0&0&0\\
0&1&0&0\\
0&0&1&0\\
0&0&0&-1
\end{array}
\right),
\hspace{3mm}
\mathrm{CS}=\left(
\begin{array}{cccc}
1&0&0&0\\
0&1&0&0\\
0&0&1&0\\
0&0&0&i
\end{array}
\right).
\]
\vspace{2mm}

\noindent
{\bf Graph-theoretic notation.}
In this work all the graphs will be undirected and unweighted. For any graph $G=(V,E)$ and any node $u\in V$, we use $N(u)$ to denote the set of neighbors of $u$.
\vspace{2mm}

\noindent
{\bf Graph states.}
Graph states are a special type of quantum states that are associated with graphs \cite{Hein+PRA04}. 
Let $G=(V,E)$ be any undirected graph. The graph state associated with $G$ is the quantum state on $|V|$ qubits constructed in the following way. Let $\{\mathsf{Q}_u\}_{u\in V}$ denote the $|V|$ registers used to store the qubits of the graph state (each register stores one qubit).
First construct the quantum state
\[
\bigotimes_{u\in V}\ket{0}_{\mathsf{Q}_u}
\]
in these registers. Then apply a Hadamard gate on each register. Finally, for each edge $\{u,v\}\in E$, apply the gate $\mathrm{CZ}$ on the couple of registers $(\mathsf{Q}_u,\mathsf{Q}_v)$. \vspace{2mm}

\noindent
{\bf The total variation distance.} Given two probability distributions $p,q\colon X\to [0,1]$ over a finite set $X$, the total variation distance (also called statistical distance) between $p$ and $q$ is defined as 
$\frac{1}{2} \sum_{x\in X}|p(x)-q(x)|$.

\subsection{Classical and quantum LOCAL models}
In this paper we consider the LOCAL communication model in both the classical and quantum scenarios. 
The topology of the network is represented by a graph. Executions proceed with round-based synchrony and each node can transfer one message to each adjacent node per round. Initially the nodes of the network share neither any randomness nor, in the quantum scenario, any entanglement.\footnote{The classical lower bound of our first result (Theorem \ref{th:relation}) actually holds even if the nodes of the network initially share arbitrary randomness.} In this paper all the networks are undirected and unweighted. All links and nodes of the network (corresponding to the edges and nodes of the graph, respectively) are reliable and suffer no faults. Each node has a distinct identifier. Initially, each node knows nothing about the topology of the network except the set of edges incident to itself and the number of nodes of the graph.

The processors at each node operate probabilistically in the classical LOCAL model, and they operate quantumly in the quantum LOCAL model. The messages exchanged between them are, respectively, classical and quantum.
In this paper we do not consider the running time of the processors, as we are only interested in the round complexity.
While the classical lower bound of Theorem~\ref{th2} is proved using a relatively informal definition of the classical LOCAL model, we include its formal definition in Appendix~\ref{sec:clasLocal} for completeness.

\subsection{The construction from prior works}\label{subseq:construction}
We now describe the construction introduced in \cite{Barrett+PRA07}, and also used in \cite{Bravyi+17}, that shows that non-locality can arise when measuring graph states.
For any even integer $d\ge 2$, we define the graph $G_d$ as a ring consisting of $3d$ nodes, and denote the nodes $v_0,v_1,\ldots,v_{3d-1}$ (see Figure \ref{fig:Gd}). It will be convenient to consider this graph as a triangle, with the three nodes $v_0$, $v_d$ and~$v_{2d}$ as corners.
We define $V_{R}=\left\{v_i\:|\:i\in\{1,\ldots,d-1\} \right\}$, $V_{B}=\left\{v_i\:|\:i\in\{d+1,\ldots,2d-1\}\right\}$ and $V_{L}=\left\{v_i\:|\:i\in\{2d+1,\ldots,3d-1\}\right\}$ as the set of nodes on the right side, bottom side and left side, respectively, of the triangle. We also define $V_{even}$ as the set of all nodes of the graph with even label, and $V_{odd}$ as the set of all nodes with odd label.

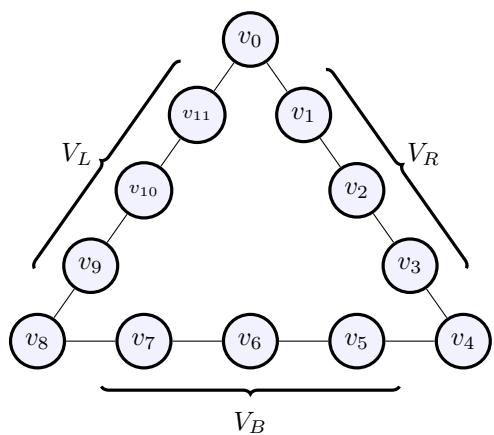
\begin{figure}[ht]
\vspace{3mm}
\centering
\begin{tikzpicture}[scale=0.4,rectnode/.style={shape=circle,draw=black,minimum height=25mm},roundnode/.style={circle, draw=black!60, fill=blue!5, very thick, minimum size=4mm}]
\newcommand\XA{7}
\newcommand\YA{10}
    
    \node[roundnode,draw=black,minimum height=6mm] (a0) at (0*\XA,1*\YA) {$v_0$};    
    \node[roundnode,draw=black,minimum height=6mm] (a1) at (0.25*\XA,0.75*\YA) {$v_1$}; 
    \node[roundnode,draw=black,minimum height=6mm] (a2) at (0.5*\XA,0.5*\YA) {$v_2$}; 
    \node[roundnode,draw=black,minimum height=6mm] (a3) at (0.75*\XA,0.25*\YA) {$v_3$}; 
    \node[roundnode,draw=black,minimum height=6mm] (a4) at (1*\XA,0*\YA) {$v_4$}; 
    \node[roundnode,draw=black,minimum height=6mm] (a5) at (0.5*\XA,0*\YA) {$v_5$};    
    \node[roundnode,draw=black,minimum height=6mm] (a6) at (0*\XA,0*\YA) {$v_6$};  
    \node[roundnode,draw=black,minimum height=6mm] (a7) at (-0.5*\XA,0*\YA) {$v_7$};  
    \node[roundnode,draw=black,minimum height=6mm] (a8) at (-1*\XA,0*\YA) {$v_8$}; 
    \node[roundnode,draw=black,minimum height=6mm] (a9) at (-0.75*\XA,0.25*\YA) {$v_9$}; 
    \node[roundnode,draw=black,minimum height=6mm] (a10) at (-0.5*\XA,0.5*\YA) {\scriptsize $v_{10}$}; 
    \node[roundnode,draw=black,minimum height=6mm] (a11) at (-0.25*\XA,0.75*\YA) {\scriptsize $v_{11}$};
    \draw[very thick,decorate,rotate=180,decoration={brace,amplitude=1mm}] (-0.7*\XA,1.5) -- (0.7*\XA,1.5);
    \draw[very thick,decorate,rotate=55,decoration={brace,amplitude=1mm}] (-0.3*\XA,7.2) -- (0.9*\XA,7.2);
    \draw[very thick,decorate,rotate=305,decoration={brace,amplitude=1mm}] (-0.85*\XA,7.2) -- (0.3*\XA,7.2);
    \node[draw=none,fill=none,] at (0,-2.7) {$V_B$};
    \node[draw=none,fill=none,] at (-5.7,6.2) {$V_L$};
    \node[draw=none,fill=none,] at (5.7,6.2) {$V_R$};

    \path[] 
     (a0) edge (a1)
     (a1) edge (a2)
     (a2) edge (a3)
     (a3) edge (a4)
     (a4) edge (a5)
     (a5) edge (a6)
     (a6) edge (a7)
     (a7) edge (a8)
     (a8) edge (a9)
     (a9) edge (a10) 
 	(a10) edge (a11) 
    (a11) edge (a0); 
 
\end{tikzpicture}
\caption{The graph $G_d$ from prior works (illustrated for $d=4$).}\label{fig:Gd}
\end{figure}

Given three bits $b_0,b_1,b_2\in\{0,1\}$, consider the process $\Pp_d(b_0,b_1,b_2)$ described in Figure~\ref{fig:process}.

\begin{figure}[h!]
\begin{center}
\fbox{
\begin{minipage}{13 cm} 
\begin{itemize}
\item[1.]
Create the graph state on the graph $G_d$.
\item[2.] 
For each $i\in\{0,1,2\}$ apply the quantum gate $\mathrm{S}^{b_i}$ to the qubit of node $v_{di}$ (i.e., depending on the value of the three bits $b_0$, $b_1$ and $b_2$, apply either the gate $\mathrm{S}$ or the identity gate $\mathrm{I}$ on each of three corner nodes $v_0$, $v_d$ and $v_{2d}$ of the graph).
\item[3.]
Apply the Hadamard gate $\mathrm{H}$ to each qubit of the graph.
\item[4.] 
Measure all qubits in the computational basis. For each $v\in V$, let $m_v$ denote the outcome of the measurement done at node $v$.
\end{itemize}
\end{minipage}
}
\end{center}\vspace{-4mm}
\caption{The process $\Pp_d(b_0,b_1,b_2)$.}\label{fig:process}
\end{figure}

From the measurement outcome of the process $\Pp_d(b_0,b_1,b_2)$, let us define four bits $m_{E}$, $m_R$, $m_B$ and $m_L$ as follows:
\begin{eqnarray*}
m_{E}&=&\bigoplus_{v\in V_{even}} m_v,\hspace{16mm}
m_{R}=\bigoplus_{v\in V_{R}\cap V_{odd}} m_v,\\
m_{B}&=&\bigoplus_{v\in V_{B}\cap V_{odd}} m_v,\hspace{12mm}
m_{L}=\bigoplus_{v\in V_{L}\cap V_{odd}} m_v.
\end{eqnarray*}
Refs.~\cite{Barrett+PRA07,Bravyi+17} characterized which combinations of these four bits can arise as an outcome of the process  $\Pp_d(b_0,b_1,b_2)$:

\begin{proposition}{(\cite{Barrett+PRA07,Bravyi+17})}\label{proposition:cond}
For any bits $b_0,b_1,b_2$ and any measurement outcome of the process $\Pp_d(b_0,b_1,b_2)$, the identity $m_{R}\oplus m_{B} \oplus m_{L} = 0$ holds. Additionally, we have:
\[
\left\{
\begin{tabular}{cc}
$m_E$&$=0 \hspace{3mm}\textrm{ if }(b_0,b_1,b_2)=(0,0,0),$\\
$m_{E}\oplus m_{R}\oplus m_{L}$&$= 1\hspace{3mm}\textrm{ if }(b_0,b_1,b_2)=(0,1,1)$,\\
$m_{E}\oplus m_{R}\oplus m_{B}$&$= 1\hspace{3mm}\textrm{ if }(b_0,b_1,b_2)=(1,0,1)$,\\
$m_{E}\oplus m_{B}\oplus m_{L}$&$= 1\hspace{3mm}\textrm{ if }(b_0,b_1,b_2)=(1,1,0)$.
\end{tabular}
\right.
\]
\end{proposition}

It will be convenient to represent a measurement outcome $\{m_v\}_{v\in V}$ as the binary string $m\in \{0,1\}^{3d}$ where the $i$-th bit is $m_{v_i}$ for each $i\in\{0,\ldots,3d-1\}$. We define the \emph{support} of the process $\Pp_d(b_0,b_1,b_2)$, and denote it $\suppP_d(b_0,b_1,b_2)$, as the set of all binary strings in $\{0,1\}^{3d}$ corresponding to measurement outcomes arising (with non-zero probability) from the process $\Pp_d(b_0,b_1,b_2)$. 


Finally, our lower bounds will rely on the following lemma, which essentially shows that the quantum correlations from the process $\Pp_d(b_0,b_1,b_2)$ cannot be simulated classically by local affine functions. 
\begin{lemma}{(\cite{Barrett+PRA07,Bravyi+17})}\label{lemma:affine}
Consider any affine function $q_{E}\colon \{0,1\}^3\to  \{0,1\}$ and any three affine functions $q_{R}\colon \{0,1\}^2\to \{0,1\}$, $q_{L}\colon \{0,1\}^2\to \{0,1\}$, $q_{B}\colon \{0,1\}^2\to \{0,1\}$ such that 
\[
q_{R}(b_0,b_1)\oplus q_{B}(b_1,b_2)\oplus q_{L}(b_0,b_2)=0
\]
holds for any $(b_1,b_2,b_3)\in \{0,1\}^3$. Then at least one of the four following equalities does not hold:
\begin{eqnarray*}
q_{E}(0,0,0) &= 0,\\
q_{E}(0,1,1) \oplus q_{R}(0,1) \oplus q_{L}(0,1)&= 1,\\
q_{E}(1,0,1) \oplus q_{R}(1,0) \oplus q_{B}(0,1)&= 1,\\
q_{E}(1,1,0) \oplus q_{B}(1,0) \oplus q_{L}(1,0)&= 1.
\end{eqnarray*}
\end{lemma}
\section{Efficient Construction of Graph States}
In this section we consider the construction of graph states in the distributed setting.  More precisely, we consider the following problem that we call the \emph{subgraph state construction problem}. The problem is defined on an arbitrary network $G=(V,E)$.
Each node $u\in V$ receives a bit $c_u\in\{0,1\}$ as input. Let $G'=(V',E')$ denote the subgraph of $G$ induced by the node set $V'=\{v\in V\:|\: c_v=1\}$. The problem asks to create the graph state corresponding to~$G'$, shared over the nodes in $V'$: each node $v\in V'$ of the network should own the corresponding 1-qubit register of the graph state (which is Register~$\mathsf{Q}_v$ in the notations of Section \ref{subsec:notation}). 


 The following theorem shows that this problem can be done efficiently.
\begin{theorem}\label{th:graph-state}
In the quantum LOCAL model, the subgraph state construction problem can be solved in 2 rounds.
\end{theorem}
\begin{proof}
The protocol is presented in Figure \ref{fig:protocol} and illustrated, for a path of two nodes, in Figure \ref{fig1}. This is clearly a 2-round protocol: one round is used at Step 1(c) and one round is used at Step 2(b).

\begin{figure}[h!]
\begin{center}
\fbox{
\begin{minipage}{13 cm} 
{\bf Input:} each node $u\in V$ receives a bit $c_u$
\vspace{2mm}
\begin{itemize}
\item[1.]
Each node $u\in V$ does the following:
\begin{itemize}
\item[(a)]
it prepares one 1-qubit register $\mathsf{Q}_u$ and, for each neighbor $v\in N(u)$, one 1-qubit register denoted $\mathsf{R}_u^v$ (all these registers are initialized to the quantum state $\ket{0}$); 
\item[(b)]
it applies a Hadamard gate on $\mathsf{Q}_u$, and then a $\mathrm{CNOT}$ gate on $(\mathsf{Q}_u,\mathsf{R}_u^v)$ with $\mathsf{Q}_u$ as control qubit, for each $v\in N(u)$;
\item[(c)]
it sends, for each $v\in N(u)$, the register $\mathsf{R}_u^v$ and the  bit $c_u$ to node $v$.
\end{itemize}
\item[2.]
Each node $u\in V$ (who now owns the registers $\mathsf{Q}_u$ and the registers $\mathsf{R}_v^u$ just received) does the following:
\begin{itemize}
\item[(a)]
it applies the gate $\mathrm{CS}$ to the pair of registers $(\mathsf{Q}_u,\mathsf{R}_v^u)$ for each $v\in N(u)$ such that $c_u\wedge c_v=1$;
\item[(b)]
it sends back Register $\mathsf{R}_v^u$ to node $v$, for each $v\in N(u)$.
\end{itemize}
\item[3.]
Each node $u\in V$ (who now owns the registers $\mathsf{Q}_u$ and the registers $\mathsf{R}_u^v$) does the following:
\begin{itemize}
\item[(a)]
it applies a $\mathrm{CNOT}$ gate on $(\mathsf{Q}_u,\mathsf{R}_u^v)$ with $\mathsf{Q}_u$ as control qubit, for each $v\in N(u)$;
\item[(b)]
it disregards the registers $\mathsf{R}_u^v$ for all $v\in N(u)$.
\end{itemize}
\end{itemize}
\end{minipage}
}
\end{center}\vspace{-4mm}
\caption{The quantum distributed algorithm solving the subgraph state construction problem.}\label{fig:protocol}
\end{figure}

We now prove that the protocol is correct. At the end of Step 1(b), the state of the whole network is:
\[
\ket{\varphi}=
\bigotimes_{u\in V}\left(\frac{1}{\sqrt{2}}\sum_{j=0}^1\left(\ket{j}_{\mathsf{Q}_u}\bigotimes_{v\in N(u)} \ket{j}_{\mathsf{R}_u^v}\right)\right).
\]
Let us fix any two nodes $u$ and $v$ such that $\{u,v\}\in E$. The state $\ket{\varphi}$ can be rewritten as
\[
\frac{1}{2}
\sum_{j=0}^1\sum_{k=0}^1\ket{j}_{\mathsf{Q}_u}\ket{k}_{\mathsf{Q}_v}
\ket{j}_{\mathsf{R}_u^v}\ket{k}_{\mathsf{R}_v^u}\ket{\psi_{j,u,k,v}}
\]
where 
\[
\ket{\psi_{j,u,k,v}}=\bigotimes_{v'\in N(u)\setminus\{v\}}\ket{j}_{\mathsf{R}_u^{v'}}\bigotimes_{u'\in N(v)\setminus\{u\}}\ket{k}_{\mathsf{R}_v^{u'}}.
\]
If node $u$ applies the gate $\mathrm{CS}$ to the pair of registers $(\mathsf{Q}_u,\mathsf{R}_v^u)$ and node $v$ applies the gate $\mathrm{CS}$ to the pair of registers $(\mathsf{Q}_v,\mathsf{R}_u^v)$, then the state $\ket{\varphi}$ is mapped to the quantum state
\[
\mathrm{CZ}_{(\mathsf{Q}_u,\mathsf{Q}_v)}\ket{\varphi},
\]
where $\mathrm{CZ}_{(\mathsf{Q}_u,\mathsf{Q}_v)}$ denotes the gate $\mathrm{CZ}$ applied to the pair of registers $(\mathsf{Q}_u,\mathsf{Q}_v)$.
Since $c_u\wedge c_v =1$ if and only if $\{u,v\}\in E'$, at the end of Step 2, the whole state of the network is 
\[
\left(
\prod_{\{u,v\}\in E'}\mathrm{CZ}_{(\mathsf{\mathsf{Q}_u},\mathsf{Q}_v)}\right)\ket{\varphi}.
\]
Since Step 3 simply disentangles and then disregards the registers $\mathsf{R}_u^v$ for all $(u,v)\in V\times V$, at the end of Step 3 we obtain the desired graph state shared by the nodes in $V'$.
\end{proof}
\begin{figure}[ht]
\vspace{3mm}
\centering
\begin{tikzpicture}[scale=0.5,rectnode/.style={shape=circle,draw=black,minimum height=25mm},roundnode/.style={circle, draw=black!60, fill=blue!5, very thick, minimum size=6mm}]
    \draw[thick] (-9.5,0) -- (9,0);
    \draw[thick] (-9.5,2) -- (-4,2);
    \draw[thick] (-9.5,4) -- (-4,4);
    \draw[thick] (-9.5,6) -- (9,6);
    
    \draw[thick] (-4,2) -- (-2,4);
    \draw[thick] (-4,4) -- (-2,2);
    
    \draw[thick] (-2,4) -- (2,4);
    \draw[thick] (-2,2) -- (2,2);
    
    \draw[thick] (2,4) -- (4,2);
    \draw[thick] (2,2) -- (4,4);
    \draw[thick] (4,2) -- (7.3,2);
    \draw[thick] (4,4) -- (7.3,4);

    \node[draw=none,fill=none,] at (-13,5) {$u$};
    \node[draw=none,fill=none,] at (-13,1) {$v$};
    \node[draw=none,fill=none,] at (-10.7,6) {$\ket{0}_{\mathsf{Q}_u}$};
    \node[draw=none,fill=none,] at (-10.7,4) {$\ket{0}_{\mathsf{R}_u^v}$};
    \node[draw=none,fill=none,] at (8.2,4) {$\ket{0}_{\mathsf{R}_u^v}$};
    \node[draw=none,fill=none,] at (-10.7,0) {$\ket{0}_{\mathsf{Q}_v}$};
    \node[draw=none,fill=none,] at (-10.7,2) {$\ket{0}_{\mathsf{R}_v^u}$};
    \node[draw=none,fill=none,] at (8.2,2) {$\ket{0}_{\mathsf{R}_v^u}$};
    \draw[thick,decorate,rotate=90,decoration={brace,amplitude=1mm}] (-0.5,12) -- (2.5,12);
    \draw[thick,decorate,rotate=90,decoration={brace,amplitude=1mm}] (3,12) -- (6.5,12);
    
    \draw[thick,decorate,rotate=0,decoration={brace,amplitude=1mm}] (-12,7.5) -- (-2,7.5);
    \draw[thick,decorate,rotate=0,decoration={brace,amplitude=1mm}] (-1.7,7.5) -- (4,7.5);
    \draw[thick,decorate,rotate=0,decoration={brace,amplitude=1mm}] (4.3,7.5) -- (9,7.5);
    
    \node[draw=none,fill=none,] at (-7,8.3) {Step 1};
    \node[draw=none,fill=none,] at (1.2,8.3) {Step 2};
    \node[draw=none,fill=none,] at (6.8,8.3) {Step 3};
    
    \node[thick,draw=black,fill=white,minimum height=4mm] (a4) at (-8,0) {$\mathrm{H}$}; 
    \node[thick,draw=black,fill=white,minimum height=4mm] (a4) at (-8,6) {$\mathrm{H}$}; 
    
    \node[thick,circle,draw=black,minimum height=5mm] (a4) at (-5,4) {}; 
    \node[circle,draw=black,fill=black,minimum height=1mm,scale=0.6] (a4) at (-5,6) {}; 
    \draw[thick] (-5,6) -- (-5,3.5);
    \node[thick,circle,draw=black,minimum height=5mm] (a4) at (-5,2) {}; 
    \node[circle,draw=black,fill=black,minimum height=1mm,scale=0.6] (a4) at (-5,0) {}; 
    \draw[thick] (-5,2.5) -- (-5,0);
    
    \node[thick,circle,draw=black,minimum height=5mm] (a4) at (5,4) {}; 
    \node[circle,draw=black,fill=black,minimum height=1mm,scale=0.6] (a4) at (5,6) {}; 
    \draw[thick] (5,6) -- (5,3.5);
    \node[thick,circle,draw=black,minimum height=5mm] (a4) at (5,2) {}; 
    \node[circle,draw=black,fill=black,minimum height=1mm,scale=0.6] (a4) at (5,0) {}; 
    \draw[thick] (5,2.5) -- (5,0);
    
    \draw [thick,draw=black,fill=white] (-0.6,-0.5) rectangle (0.6,2.5);
    \node[draw=none,fill=none,] at (0,1) {$\mathrm{CS}$}; 
    \draw [thick,draw=black,fill=white] (-0.6,3.5) rectangle (0.6,6.5);
    \node[draw=none,fill=none,] at (0,5) {$\mathrm{CS}$}; 

\end{tikzpicture}
\caption{Our protocol illustrated for a 2-path graph $G=(V,E)$ with $V=\{u,v\}$, $E=\{\{u,v\}\}$ and $c_u=c_v=1$ (the classical messages are omitted from the figure).}
 \label{fig1}
\end{figure}
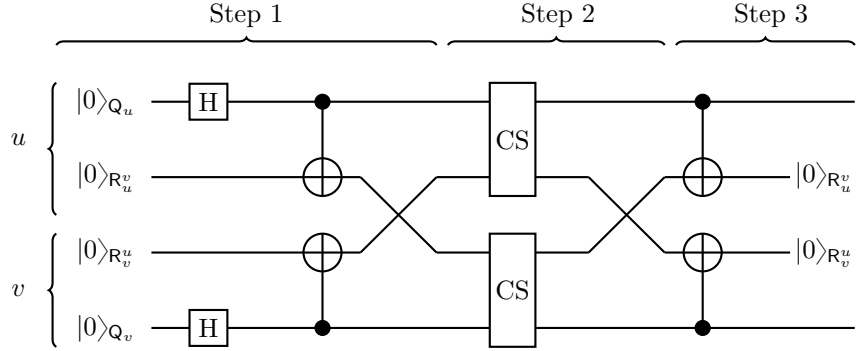

\section{Separation between the Classical and Quantum LOCAL Models} \label{sec:RelSep}
In this section we prove Theorem \ref{th1}.

For any even integer $d\ge 2$, remember the network $G_d=(V,E)$ defined in Section \ref{subseq:construction}, where $V=\{v_0,\ldots,v_{3d-1}\}$. In this section we will consider the network~$\Gg_d$ with node set $V\cup\{w_0,w_1,w_2\}$ and edge set $E\cup\{\{v_0,w_0\},\{v_d,w_1\},\{v_{2d},w_2\}\}$, which is obtained from $G_d$ by adding one node to each corner (see Figure~\ref{fig:graph2}).

\begin{figure}[ht]
\vspace{3mm}
\centering
\begin{tikzpicture}[scale=0.4,rectnode/.style={shape=circle,draw=black,minimum height=25mm},roundnode/.style={circle, draw=black!60, fill=blue!5, very thick, minimum size=4mm}]
\newcommand\XA{7}
\newcommand\YA{10}
    
    \node[roundnode,draw=black,minimum height=6mm] (b0) at (0*\XA,1.3*\YA) {\footnotesize $w_0$};  
    \node[roundnode,draw=black,minimum height=6mm] (b2) at (-1.4*\XA,-0.15*\YA) {\footnotesize $w_2$};  
    \node[roundnode,draw=black,minimum height=6mm] (b1) at (1.4*\XA,-0.15*\YA) {\footnotesize $w_1$};
    \node[roundnode,draw=black,minimum height=6mm] (a0) at (0*\XA,1*\YA) {$v_0$};    
    \node[roundnode,draw=black,minimum height=6mm] (a1) at (0.25*\XA,0.75*\YA) {$v_1$}; 
    \node[roundnode,draw=black,minimum height=6mm] (a2) at (0.5*\XA,0.5*\YA) {$v_2$}; 
    \node[roundnode,draw=black,minimum height=6mm] (a3) at (0.75*\XA,0.25*\YA) {$v_3$}; 
    \node[roundnode,draw=black,minimum height=6mm] (a4) at (1*\XA,0*\YA) {$v_4$}; 
    \node[roundnode,draw=black,minimum height=6mm] (a5) at (0.5*\XA,0*\YA) {$v_5$};    
    \node[roundnode,draw=black,minimum height=6mm] (a6) at (0*\XA,0*\YA) {$v_6$};  
    \node[roundnode,draw=black,minimum height=6mm] (a7) at (-0.5*\XA,0*\YA) {$v_7$};  
    \node[roundnode,draw=black,minimum height=6mm] (a8) at (-1*\XA,0*\YA) {$v_8$}; 
    \node[roundnode,draw=black,minimum height=6mm] (a9) at (-0.75*\XA,0.25*\YA) {$v_9$}; 
    \node[roundnode,draw=black,minimum height=6mm] (a10) at (-0.5*\XA,0.5*\YA) {\scriptsize $v_{10}$}; 
    \node[roundnode,draw=black,minimum height=6mm] (a11) at (-0.25*\XA,0.75*\YA) {\scriptsize $v_{11}$};
    \draw[very thick,decorate,rotate=180,decoration={brace,amplitude=1mm}] (-0.7*\XA,1.5) -- (0.7*\XA,1.5);
    \draw[very thick,decorate,rotate=55,decoration={brace,amplitude=1mm}] (-0.3*\XA,7.2) -- (0.9*\XA,7.2);
    \draw[very thick,decorate,rotate=305,decoration={brace,amplitude=1mm}] (-0.85*\XA,7.2) -- (0.3*\XA,7.2);
    \node[draw=none,fill=none,] at (0,-2.7) {$V_B$};
    \node[draw=none,fill=none,] at (-5.7,6.2) {$V_L$};
    \node[draw=none,fill=none,] at (5.7,6.2) {$V_R$};

    \path[] 
     (a0) edge (b0)
     (a4) edge (b1)
     (a8) edge (b2)
     (a0) edge (a1)
     (a1) edge (a2)
     (a2) edge (a3)
     (a3) edge (a4)
     (a4) edge (a5)
     (a5) edge (a6)
     (a6) edge (a7)
     (a7) edge (a8)
     (a8) edge (a9)
     (a9) edge (a10) 
 	(a10) edge (a11) 
    (a11) edge (a0); 
 
\end{tikzpicture}
\caption{The network $\Gg_d$ considered to prove the separation (illustrated for $d=4$).}\label{fig:graph2}
\end{figure}
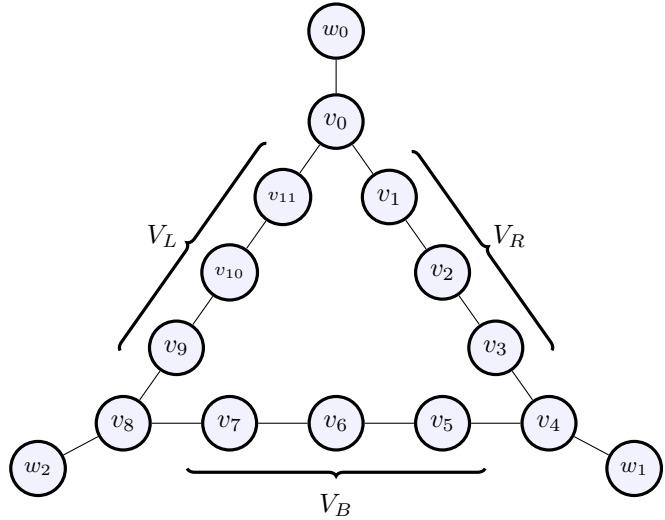

We now describe the computational problem used to prove our separation.
 The network considered is $\Gg_d$, for any even integer $d\ge 2$. 
 The input consists of three bits $b_0$, $b_1$ and $b_2$: node $w_0$ is given $b_0$, node $w_1$ is given $b_1$, and node $w_{2}$ is given $b_2$ (the other nodes have no input). The output is defined as follows: for each $i\in\{0,1,\ldots,3d-1\}$, the node $v_i$ should output one bit $x_i$. The nodes $w_0$, $w_1$ and $w_2$ do not output anything. The output can thus be seen as a binary string $(x_0,\ldots,x_{3d-1})$ of length $3d$. We say that this string is \emph{valid} if it is in the set $\suppP_d(b_0,b_1,b_2)$. 

The following theorem shows an upper bound on the complexity of this problem in the quantum LOCAL model and a lower bound in the classical LOCAL model.
\begin{theorem}\label{th:relation}
There exists a 2-round quantum algorithm that always outputs a valid string. For any integer $T\le d/2$, no $T$-round classical algorithm can output a valid string with probability greater than $7/8$ on all inputs $(b_0,b_1,b_2)\in\{0,1\}^3$, even if arbitrary prior randomness is allowed.
\end{theorem}
\begin{proof}
The considered computational problem can easily be solved in two rounds in the quantum setting by implementing the following process.
\begin{process} \label{proc:rel}
The nodes of the network first apply the 2-round algorithm of Theorem \ref{th:graph-state} with input $c_{w_0}=c_{w_1}=c_{w_2}=0$ and  $c_{v}=1$ for each $v\in V$.
This constructs the graph state over the subgraph $G_d$ of $\Gg_d$. Moreover, for each $i\in\{0,1,2\}$, the node $w_i$ concurrently sends its input $b_i$ to its neighbor $v_{di}$ (the messages can be appended to the messages of the algorithm of Theorem \ref{th:graph-state}). Finally, the nodes of $V$ implement Steps~2--4 of the process $\Pp_d(b_0,b_1,b_2)$, which can be done without communication, and output their measurement outcomes.
\end{process}
Note that implementing Process~\ref{proc:rel} requires each node to know whether it is an input node ($w_0$, $w_1$ or $w_2$), a corner node on the ring ($v_0$, $v_d$ or $v_{2d}$) or a non-corner node on the ring (all the other nodes). This is not a problem since each node knows its degree and the type of the nodes depends only on their degrees: the nodes $w_0$, $w_1$ and $w_2$ are the nodes of degree 1, the nodes $v_0$, $v_d$ and $v_{2d}$ are the nodes of degree 3, and all the other nodes have degree 2.

We now show the classical lower bound, which uses the same argument as in \cite{Barrett+PRA07} and holds even if the nodes of the network share prior randomness. Consider any classical distributed algorithm~$\mathscr{A}$ and fix its randomness $r$ (the string $r$ represents both the shared prior randomness and the random bits used by the algorithm). This defines a deterministic algorithm that we denote $\mathscr{A}(r)$. Let us write $q_v(b_0,b_1,b_2)$ the bit output at node $v$ by $\mathscr{A}(r)$, for each $v\in V$. Let us define 
\begin{eqnarray*}
q_{E}(b_0,b_1,b_2)&=&\bigoplus_{v\in V_{even}} q_v(b_0,b_1,b_2),\hspace{16mm}
q_{R}(b_0,b_1,b_2)=\bigoplus_{v\in V_{R}\cap V_{odd}} q_v(b_0,b_1,b_2),\\
q_{B}(b_0,b_1,b_2)&=&\bigoplus_{v\in V_{B}\cap V_{odd}} q_v(b_0,b_1,b_2),\hspace{12mm}
q_{L}(b_0,b_1,b_2)=\bigoplus_{v\in V_{L}\cap V_{odd}} q_v(b_0,b_1,b_2).
\end{eqnarray*}
Assume that the algorithm uses at most $d/2$ rounds. Then $q_E$,  $q_{R}$, $q_B$ and $q_L$ are affine functions of $b_0,b_1,b_2$. Moreover, $q_{R}$ can only depend on $b_0$ and $b_1$, $q_{B}$ can only depend on $b_1$ and $b_2$, and $q_{L}$ can only depend on $b_0$ and $b_2$. From Proposition~\ref{proposition:cond} and Lemma~\ref{lemma:affine} we get that, at least for one choice of $(b_0,b_1,b_2)\in\{0,1\}^3$, the output of $\mathscr{A}(r)$ is not a valid string (i.e., does not correspond to a possible measurement outcome of the process $\Pp_d(b_0,b_1,b_2)$). A simple counting argument then shows that there exists at least one choice of $(b_0,b_1,b_2)$ for which the original randomized protocol $\mathscr{A}$ fails to output a valid string with probability at least $1/8$.
\end{proof}
We are now ready to prove Theorem \ref{th1}.
\begin{proof}[Proof of Theorem \ref{th1}]
Theorem \ref{th:relation} implies that any classical algorithm that outputs a valid string with probability greater than $7/8$ requires a number of rounds linear in the size of the network (since $d$ is a linear function of the size of network $\Gg_d$). 

We now show how to reduce the success probability from $7/8$ to an arbitrary small value: for any constant $\varepsilon> 0$ we construct a new computational problem, which can still be solved in two rounds in the quantum setting, such that any classical algorithm solving this problem with probability at least $\varepsilon$ requires a number of rounds linear in the size of the network. Let $k$ be an integer. The problem considered is simply $k$ independent copies of the problem considered so far: the network considered has $3k(d+1)$ nodes and consists of $k$ copies of the network $\Gg_d$. Each copy receives three bits and outputs a string of $3d$ bits. The output of the whole network is correct if the strings output by each copy are all valid. This problem can obviously be solved using two rounds in the quantum setting by constructing the graph state over the whole network. Theorem \ref{th:relation} implies that for any integer $T\le d/2$, no $T$-round classical algorithm can give a correct output with probability greater than $(7/8)^k$ on all inputs, even if arbitrary prior randomness is allowed. Setting $k=\Theta(\log (1/\varepsilon))$ concludes the proof.
\end{proof}

\section{Separation for a Distribution} \label{sec:DistSep}
In this section we prove Theorem \ref{th2}. The idea is to convert the relation of the previous section into a distribution by requiring that each input is taken uniformly at random (and requiring that the three nodes with an input output their inputs as well).

Recall Process~\ref{proc:rel} in the proof of Theorem~\ref{th:relation}. There, the actions of every node of $\Gg_d$ depend only on the degree of the node, namely, whether its degree is $1$, $2$ or $3$. The same is true for the $2$-round sampling protocol in the quantum LOCAL model described below, which also uses the same network $\Gg_d$. Therefore, for notational convenience, let us assume that every node knows its global location in $\Gg_d$.

Consider the probability distribution $\distrP_d$ generated by the following $2$-round quantum protocol.
First, for each $i\in\{0,1,2\}$, the node $w_i$ chooses an unbiased random bit $b_i$.
Then Process~\ref{proc:rel} is implemented, at the end of which, as specified, nodes $u\in V$ each return one bit.  Meanwhile, the nodes $w_0,w_1,w_2$ output, respectively, $b_0,b_1,b_2$.

Theorem \ref{th2} immediately follows from the following result.
\begin{theorem}
Every $T\le d/4$ round algorithm on $\Gg_d$ in the classical LOCAL model generates a probability distribution that is at least  $1/11$ away from $\distrP_d$  in the total variation distance.
 \end{theorem}
\begin{proof} 
The proof proceeds as follows. Starting from the classical LOCAL model, we present a series of increasingly powerful models on the network $\Gg_d$. Each model receives no input and returns one bit per node. Then we show that the last, the most powerful among these models cannot generate a probability distribution that is closer than distance $1/11$ away from $\distrP_d$.

Consider the classical LOCAL model on the network $\Gg_d$.
We assume that the randomness of each node stems from a finite random bit string that it receives as an input, and all further operations of the node are deterministic (see Appendix~\ref{sec:FiniteRand} for technical details).
We now present a series of steps, where each step either strengthens the model or maintains its power while making it easier to analyze.
\begin{enumerate}
\item
We assume that all the nodes know their location in the global topology.
\item
We allow certain nodes to share randomness. In particular, for each $i\in\{0,1,2\}$, let $V_i$ be the set consisting of $w_{i}$ and all the nodes $u\in V$ at distance at most $T$ away from $w_{i}$.
And let $V_{\perp}=V\setminus (V_0\cup V_1\cup V_2)$.
We assume that, for $i\in\{0,1,2,\perp\}$, all nodes within $V_i$ share randomness, namely, they all start with the same random string $Q_i$, which we think of as a random variable.
\end{enumerate}

Here it is worth pausing the steps to note that, in a $T$-round protocol, the bit $b_i$ output by the node $w_{i}$ depends only on $Q_i$, thus we may write it as a function $b_i(Q_i)$. Let $p_i$ be the probability that $b_i=1$.
If there exists $i\in\{0,1,2\}$ with $p_i\notin[5/11,6/11]$, then the marginal distribution over $b_i$ is already at total variation distance greater than $1/11$ away from the corresponding marginal distribution in $\distrP_d$, and the whole distributions ($\distrP_d$ and the one generated by the classical protocol) can be only even farther apart. Thus let us assume that $p_i\in[5/11,6/11]$ for all $i\in\{0,1,2\}$. Since $Q_0,Q_1,Q_2$ are independent, each $(b_0,b_1,b_2)\in\{0,1\}^3$ is output with probability at least $(5/11)^3>1/11$.

\begin{enumerate}[resume]
\item
For $i\in\{0,1,2\}$, let $B_i$ be a random variable that takes value $1$ with probability $p_i$ and value $0$ with probability $1-p_i$. For both $\beta\in\{0,1\}$, let $Q_i^\beta$ be a random variable
 that equals each value $q$ of $Q_i$ such that $b_i(q)=\beta$ with probability $\Pr[Q_i=q]/\Pr[B_i=\beta]$.
We replace the shared randomness $Q_i$ by $(Q_i^0,  Q_i^1, B_i)$---each of the three variables being independent---with an extra requirement that the node $w_{i}$ always outputs $B_i$.
This is clearly without loss of power, because we can recover $Q_i$ as $Q_i^{B_i}$, for which $b_i(Q_i)=B_i$.
\item \label{model:powerful}
We share all the randomness except $B_0,B_1,B_2$ among all the nodes. More precisely, we assume that all nodes start with the randomness
$r=(Q_\perp,Q_0^0,Q_0^1,Q_1^0,Q_1^1,Q_2^0,Q_2^1)$.
 In addition, for each $i\in\{0,1,2\}$, nodes in $V_i$ start with an additional random bit $B_i$ and we preserve the requirement that $w_{i}$ must output $B_i$.
\end{enumerate}

Now we need to show that the final model cannot generate a probability distribution that is closer than total variation distance $1/11$ away from $\distrP_d$. Note that, at the beginning of the protocol, the value $B_i$ is only known to nodes at distance at most $T-1$ away from $v_{di}$, and, after the protocol, it can be known only to nodes at distance $2T-1<d/2$ away from $v_{di}$. In particular, at the end of the protocol, each node of the network will know no more than one of the values $B_0,B_1,B_2$. All other communicated information is useless, as, aside from $B_0,B_1,B_2$, all other randomness is global.

The remainder of the proof is almost equivalent to that of the classical lower bound in Theorem~\ref {th:relation}, with the sole difference of the counting argument: instead of each choice of $(b_0,b_1,b_2)$ being given with probability exactly $1/8$, now each choice of $(b_0,b_1,b_2)$ is given with probability at least $1/11$.
%
%
\end{proof}

\section*{Acknowledgments}
The authors are grateful to Richard Cleve, Hirotada Kobayashi and Tomoyuki Morimae for helpful comments and discussions. 
FLG was partially supported by the JSPS KAKENHI grants No.~15H01677, No.~16H01705 and No.~16H05853. HN was partially supported by the JSPS KAKENHI grants No.~26247016, No.~16H01705 and No.~16K00015.
AR was partially supported by the Singapore Ministry of Education and the National Research Foundation under grant R-710-000-012-135. 
Part of this work was done while AR was visiting Kyoto University, and AR would like to thank FLG for hospitality.

\begin{thebibliography}{10}

\bibitem{Aaronson+05}
Scott Aaronson and Andris Ambainis.
\newblock Quantum search of spatial regions.
\newblock {\em Theory of Computing}, 1(1):47--79, 2005.
\newblock \href {http://dx.doi.org/10.4086/toc.2005.v001a004}
  {\path{doi:10.4086/toc.2005.v001a004}}.

\bibitem{Arfaoui+14}
Heger Arfaoui and Pierre Fraigniaud.
\newblock What can be computed without communications?
\newblock {\em {SIGACT} News}, 45(3):82--104, 2014.
\newblock \href {http://dx.doi.org/10.1145/2670418.2670440}
  {\path{doi:10.1145/2670418.2670440}}.

\bibitem{Barrett+PRA07}
Jonathan Barrett, Carlton~M. Caves, Bryan Eastin, Matthew~B. Elliott, and
  Stefano Pironio.
\newblock Modeling {Pauli} measurements on graph states with nearest-neighbor
  classical communication.
\newblock {\em Physical Review A}, 75:012103, 2007.
\newblock \href {http://dx.doi.org/10.1103/PhysRevA.75.012103}
  {\path{doi:10.1103/PhysRevA.75.012103}}.

\bibitem{Ben-Or+STOC05}
Michael Ben{-}Or and Avinatan Hassidim.
\newblock Fast quantum byzantine agreement.
\newblock In {\em Proceedings of the 37th Annual {ACM} Symposium on Theory of
  Computing}, pages 481--485, 2005.
\newblock \href {http://dx.doi.org/10.1145/1060590.1060662}
  {\path{doi:10.1145/1060590.1060662}}.

\bibitem{Bravyi+17}
Sergey Bravyi, David Gosset, and Robert K\"onig.
\newblock Quantum advantage with shallow circuits.
\newblock arXiv:1704.00690. Presented as an invited talk at STOC 2018.

\bibitem{Broadbent+08}
Anne Broadbent and Alain Tapp.
\newblock Can quantum mechanics help distributed computing?
\newblock {\em SIGACT News}, 39(3):67--76, 2008.
\newblock \href {http://dx.doi.org/10.1145/1412700.1412717}
  {\path{doi:10.1145/1412700.1412717}}.

\bibitem{Buhrman+STOC98}
Harry Buhrman, Richard Cleve, and Avi Wigderson.
\newblock Quantum vs. classical communication and computation.
\newblock In {\em Proceedings of the 30th {ACM} Symposium on the Theory of
  Computing}, pages 63--68, 1998.
\newblock \href {http://dx.doi.org/10.1145/276698.276713}
  {\path{doi:10.1145/276698.276713}}.

\bibitem{Chang+FOCS17}
Yi{-}Jun Chang and Seth Pettie.
\newblock A time hierarchy theorem for the {LOCAL} model.
\newblock In {\em Proceedings of the 58th {IEEE} Annual Symposium on
  Foundations of Computer Science}, pages 156--167, 2017.
\newblock \href {http://dx.doi.org/10.1109/FOCS.2017.23}
  {\path{doi:10.1109/FOCS.2017.23}}.

\bibitem{Denchev+08}
Vasil~S. Denchev and Gopal Pandurangan.
\newblock Distributed quantum computing: a new frontier in distributed systems
  or science fiction?
\newblock {\em SIGACT News}, 39(3):77--95, 2008.
\newblock \href {http://dx.doi.org/10.1145/1412700.1412718}
  {\path{doi:10.1145/1412700.1412718}}.

\bibitem{Elkin+PODC14}
Michael Elkin, Hartmut Klauck, Danupon Nanongkai, and Gopal Pandurangan.
\newblock Can quantum communication speed up distributed computation?
\newblock In {\em Proceedings of the 2014 {ACM} Symposium on Principles of
  Distributed Computing}, pages 166--175, 2014.
\newblock \href {http://dx.doi.org/10.1145/2611462.2611488}
  {\path{doi:10.1145/2611462.2611488}}.

\bibitem{Frischknecht+SODA12}
Silvio Frischknecht, Stephan Holzer, and Roger Wattenhofer.
\newblock Networks cannot compute their diameter in sublinear time.
\newblock In {\em Proceedings of the 23rd Annual {ACM-SIAM} Symposium on
  Discrete Algorithms}, pages 1150--1162, 2012.
\newblock \href {http://dx.doi.org/10.1137/1.9781611973099.91}
  {\path{doi:10.1137/1.9781611973099.91}}.

\bibitem{Gavoille+DISC09}
Cyril Gavoille, Adrian Kosowski, and Marcin Markiewicz.
\newblock What can be observed locally?
\newblock In {\em Proceedings of the 23rd International Symposium on
  Distributed Computing}, pages 243--257, 2009.
\newblock \href {http://dx.doi.org/10.1007/978-3-642-04355-0_26}
  {\path{doi:10.1007/978-3-642-04355-0_26}}.

\bibitem{Greenberger+89}
Daniel~M. Greenberger, Michael~A. Horne, and Anton Zeilinger.
\newblock Going beyond {Bell}’s theorem.
\newblock In {\em Bell’s Theorem, Quantum Theory, and Conceptions of the
  Universe}, volume~37 of {\em Fundamental Theories of Physics}, pages 69--72.
  Springer, Dordrecht, 1989.
\newblock \href {http://dx.doi.org/10.1007/978-94-017-0849-4_10}
  {\path{doi:10.1007/978-94-017-0849-4_10}}.

\bibitem{Arrow+Nature17}
Aram~W. Harrow and Ashley Montanaro.
\newblock Quantum computational supremacy.
\newblock {\em Nature}, 549:203–209, 2017.
\newblock \href {http://dx.doi.org/10.1038/nature23458}
  {\path{doi:10.1038/nature23458}}.

\bibitem{Hein+PRA04}
Marc Hein, Jens Eisert, and Hans~J. Briegel.
\newblock Multiparty entanglement in graph states.
\newblock {\em Physical Review A}, 69:062311, Jun 2004.
\newblock \href {http://dx.doi.org/10.1103/PhysRevA.69.062311}
  {\path{doi:10.1103/PhysRevA.69.062311}}.

\bibitem{Holzer+PODC12}
Stephan Holzer and Roger Wattenhofer.
\newblock Optimal distributed all pairs shortest paths and applications.
\newblock In {\em Proceedings of the 2012 ACM Symposium on Principles of
  Distributed Computing}, pages 355--364, 2012.
\newblock \href {http://dx.doi.org/10.1145/2332432.2332504}
  {\path{doi:10.1145/2332432.2332504}}.

\bibitem{Hoyer+STACS02}
Peter H{\o}yer and Ronald de~Wolf.
\newblock Improved quantum communication complexity bounds for disjointness and
  equality.
\newblock In {\em Proceedings of the 19th Annual Symposium on Theoretical
  Aspects of Computer Science}, pages 299--310, 2002.
\newblock \href {http://dx.doi.org/10.1007/3-540-45841-7_24}
  {\path{doi:10.1007/3-540-45841-7_24}}.

\bibitem{LeGall+PODC18}
Fran{\c{c}}ois {Le Gall} and Fr{\'e}d{\'e}ric Magniez.
\newblock Sublinear-time quantum computation of the diameter in {CONGEST}
  networks.
\newblock In {\em Proceedings of the 2018 {ACM} Symposium on Principles of
  Distributed Computing}, pages 337--346, 2018.
\newblock \href {http://dx.doi.org/10.1145/3212734.3212744}
  {\path{doi:10.1145/3212734.3212744}}.

\bibitem{LinialFOCS87}
Nathan Linial.
\newblock Distributive graph algorithms-global solutions from local data.
\newblock In {\em Proceedings of the 28th Annual Symposium on Foundations of
  Computer Science}, pages 331--335, 1987.
\newblock \href {http://dx.doi.org/10.1109/SFCS.1987.20}
  {\path{doi:10.1109/SFCS.1987.20}}.

\bibitem{LinialSICOMP92}
Nathan Linial.
\newblock Locality in distributed graph algorithms.
\newblock {\em {SIAM} Journal on Computing}, 21(1):193--201, 1992.
\newblock \href {http://dx.doi.org/10.1137/0221015}
  {\path{doi:10.1137/0221015}}.

\bibitem{Nielsen+00}
Michael~A. Nielsen and Isaac~L. Chuang.
\newblock {\em Quantum Computation and Quantum Information}.
\newblock Cambridge University Press, 2011.
\newblock \href {http://dx.doi.org/10.1017/CBO9780511976667}
  {\path{doi:10.1017/CBO9780511976667}}.

\bibitem{Peleg00}
David Peleg.
\newblock {\em Distributed computing: a locality-sensitive approach}.
\newblock Society for Industrial and Applied Mathematics, 2000.
\newblock \href {http://dx.doi.org/10.1137/1.9780898719772}
  {\path{doi:10.1137/1.9780898719772}}.

\bibitem{Peleg+ICALP12}
David Peleg, Liam Roditty, and Elad Tal.
\newblock Distributed algorithms for network diameter and girth.
\newblock In {\em Proceedings of the 39th International Colloquium on Automata,
  Languages, and Programming}, pages 660--672, 2012.
\newblock \href {http://dx.doi.org/10.1007/978-3-642-31585-5_58}
  {\path{doi:10.1007/978-3-642-31585-5_58}}.

\bibitem{Tani+12}
Seiichiro Tani, Hirotada Kobayashi, and Keiji Matsumoto.
\newblock Exact quantum algorithms for the leader election problem.
\newblock {\em {ACM} Transactions on Computation Theory}, 4(1):1:1--1:24, 2012.
\newblock \href {http://dx.doi.org/10.1145/2141938.2141939}
  {\path{doi:10.1145/2141938.2141939}}.

\bibitem{deWolf02}
{Ronald de} Wolf.
\newblock Quantum communication and complexity.
\newblock {\em Theoretical Computer Science}, 287(1):337--353, 2002.
\newblock \href {http://dx.doi.org/10.1016/S0304-3975(02)00377-8}
  {\path{doi:10.1016/S0304-3975(02)00377-8}}.

\end{thebibliography}

\appendix

\section{Technical Definition of the Classical LOCAL Model}\label{sec:clasLocal}

We formalize a $T$-round classical LOCAL network as follows.
We model each node $u\in V$ as a special Turing machine with a work
tape, a message tape $M_{u,v}$ for each neighbor $v\in N(u)$, and a
read-only random tape. Initially, the work tape contains the input of $u$
(if there is any), the message tapes are blank, and the random tape is initialized to unbiased random
bits, independent from one another and from the content of other tapes.

The set of states of each Turing machine is a disjoint union $S_0 \cup
\ldots \cup S_T \cup \{q_{\mathrm{fin}}\}$, with one designated
``starting'' state $q_t\in S_t$ for each $t\in\{0,\ldots,T\}$. The state
$q_{\mathrm{fin}}$ is the final state, and, for convenience, we define
$q_{T+1}=q_{\mathrm{fin}}$. The Turing machine starts in $q_0$, and,
 for every $t\in\{0,\ldots,T\}$,
 we require that a state in $S_t$ can only transition into
a state in $S_t \cup \{q_{t+1}\}$. In addition, we require that the
transition from $S_t$ to $q_{t+1}$ occurs with probability $1$, regardless of
the content of the work and the message tapes when the Turing
machine first enters $q_t$.

We formalize the exchange of messages as follows.
In round $t\in\{0,\ldots,T\}$, all Turing machines start in their
corresponding state $q_t$ and run until they all have reached their
corresponding state $q_{t+1}$. Then, if $t<T$, the configuration of
message tapes $M_{u,v}$ and $M_{v,u}$ are swapped for every $\{u,v\}\in E$,
and all Turing machines start round $t+1$. Otherwise, if $t=T$, the work tape of $u\in V$ contains the output of that node.

\subsection{Restriction to finite and initial randomness} \label{sec:FiniteRand}

In the proof of Theorem~\ref{th2}, we are essentially assuming that the random tapes are of finite length. That is without loss of generality because, given any protocol on a finite network and any $\epsilon>0$, there exists a positive integer $L$ such that, with probability at least $1-\epsilon$, no Turing machine of the protocol ever visits more than $L$ cells of its random tape.
Thus, since $\epsilon$ can be chosen arbitrarily small, we can assume all random tapes to be of some finite length $L$. Via similar reasoning, we can assume that all the randomness is provided at the beginning of the protocol, instead of fresh randomness being provided at each round.

\section{The Case of Functions}\label{sec:functions}
A well-known fact in classical distributed computing is that randomness does not help when computing functions in the LOCAL model. In this appendix we show that this argument extends to the quantum case: we prove that any $T$-round quantum protocol computing a function can be converted into a $T$-round classical protocol computing the same function.

Suppose, in the LOCAL model, we have a $T$-round quantum protocol $\Pp$ with the network structure given by a graph $G=(V,E)$.
And suppose that $\Pp$ computes some function $f\colon D \rightarrow \Sigma^{|V|}$, where $\Sigma$ is the input-output alphabet and $D\subseteq\Sigma^{|V|}$.
More precisely, we assume that, for every input $x\in D$, with probability strictly larger than $1/2$ all nodes $u\in V$ output $f(x)_u$.

For a node $u\in V$ and an integer $i\ge 0$, let the \emph{$i$-neighborhood} of $u$, denoted $N_i(u)$, be the set of nodes in $V$ at distance at most $i$ away from $u\in V$. 
And, for an input $x\in D$, let $x_{u,i}$ denote the restriction of $x$ to $N_i(u)$.

\begin{claim}\label{clm:TNeighbor}
For every $x\in D$ and every $u\in V$, the output of node $u$ is a random variable $O_u(x)$ whose probability distribution depends only on $x_{u,T}$.
(This holds true even in a more powerful model where nodes are allowed to share any entanglement prior to receiving the input.)
\end{claim}

Since the quantum protocol $\Pp$ computes $f$, for every $x\in D$ and every $u\in V$, the random variable $O_u(x)$ takes the value $f(x)_u$ with probability larger than $1/2$. 
Now consider the following classical $T$-round \emph{deterministic} protocol
: each node $u\in V$ collects the inputs from nodes in its $T$-neighborhood, which suffices to locally reproduce $O_u(x)$, and then it outputs the most probable value of $O_u(x)$.
The correctness of this protocol follows from Claim~\ref{clm:TNeighbor}.

\begin{proof}[Proof of Claim~\ref{clm:TNeighbor}]
For $t\in\{0,1,\ldots,T\}$, let $\rho_t$ be the reduced density state of the $(T-t)$-neighborhood of $u$ after $t$ rounds of communication. By induction, we argue that the states $\rho_0,\rho_1,\ldots,\rho_T$---which we can think of forming the past light cone of $\rho_T$---all depend only on $x_{u,T}$, and no values of $x$ outside $N_T(u)$.
As the base case, it clearly holds for $\rho_0$ (even in the presence of prior entanglement).
For the inductive step, let us assume that, for some $t\ge 0$, $\rho_t$ depends only on $x_{u,T}$.
Then the reduced density state of the $(T-t)$-neighborhood of $u$ just before the $(t+1)$-th round of communication depends only on $x_{u,T}$.
In that round of communication, nodes in the $(T-t-1)$-neighborhood of $u$ receive messages only from within the $(T-t)$-neighbourhood of $u$, 
and thus the state $\rho_{t+1}$ also depends only on $x_{u,T}$.
When $\rho_T$, the final state of the node $u$, is measured, the probabilities of various outcomes are completely determined by $\rho_T$. Hence, these probabilities depend only on $x_{u,T}$.
\end{proof}

\end{document}